\documentclass{article}
\usepackage[utf8]{inputenc}

\usepackage[english]{babel}

\usepackage{float}
\usepackage{fullpage}
\usepackage{xcolor}
\usepackage[backend=biber,style=alphabetic,maxbibnames=9,maxalphanames=9,maxcitenames=9,safeinputenc]{biblatex}
\usepackage{iftex}
\ifPDFTeX
   \usepackage[utf8]{inputenc}
   \usepackage[T1]{fontenc}
\else
   \ifXeTeX
     \usepackage{fontspec}
   \else
     \usepackage{luatextra}
   \fi
   \defaultfontfeatures{Ligatures=TeX}
\fi
\usepackage{csquotes}

\usepackage{amsfonts,amsmath,amssymb,mathtools,amsthm}

\usepackage{mleftright}\mleftright

\usepackage{bm}

\usepackage{aliascnt} 
\usepackage[final]{hyperref}

\usepackage{bookmark} 

\usepackage{newunicodechar} 

\usepackage{xparse}

\addto\extrasenglish{%
  \def\equationautorefname~#1\null{Equation~(#1)\null}
}

\newtheorem{theorem}{Theorem}

\newcommand{\addalias}[3][theorem]{
  \newaliascnt{#2}{#1}
  \newtheorem{#2}[#2]{#3}
  \aliascntresetthe{#2}
  \expandafter\def\csname #2autorefname\endcsname{#3}
}
\addalias{lemma}{Lemma}
\addalias{claim}{Claim}
\addalias{corollary}{Corollary}
\addalias{conjecture}{Conjecture}

\addalias{definition}{Definition}
\addalias{problem}{Problem}

\floatstyle{boxed}

\newfloat{algorithm}{t}{lop}
\floatname{algorithm}{Algorithm}

\DeclareBibliographyDriver{arxiv}{%
  \usebibmacro{bibindex}%
  \usebibmacro{begentry}%
  \usebibmacro{author/translator+others}%
  \setunit{\printdelim{nametitledelim}}\newblock
  \usebibmacro{title}%
  \newunit
  \printlist{language}%
  \newunit\newblock
  \usebibmacro{byauthor}%
  \newunit\newblock
  \usebibmacro{date}%
  \newunit\newblock
  \printfield{note}%
  \newunit\newblock
  \usebibmacro{doi+eprint+url}%
  \newunit\newblock
  \usebibmacro{addendum+pubstate}%
  \setunit{\bibpagerefpunct}\newblock
  \usebibmacro{pageref}%
  \newunit\newblock
  \iftoggle{bbx:related}
    {\usebibmacro{related:init}%
     \usebibmacro{related}}
    {}%
  \usebibmacro{finentry}}

\DeclareFieldFormat[arxiv]{title}{\mkbibquote{#1\isdot}}

\def\01{\{0,1\}}

\newcommand{\trm}[1]{\left(#1\right)}

\newcommand{\ceil}[1]{\left\lceil#1\right\rceil}

\newcommand{\eps}{\varepsilon}
\renewcommand{\epsilon}{\varepsilon}
\newcommand{\ii}{\mathbf{i}} 

\let\ex = \expandafter

\makeatletter
\let\ifnextchar\@ifnextchar
\makeatother

\newcommand{\newfunction}[2]{
\DeclareDocumentCommand{#1}{e_e^d()}{
{#2}%
\IfValueT{##1}{%
  _{##1}%
  }%
 \IfValueT{##2}{%
  ^{##2}%
  } %
\IfValueT{##3}{\trm{##3}}%
}
}

\newcommand{\newtextfunction}[2]{\newfunction{#1}{\operatorname{#2}}}

\newcommand{\newlistcommand}[4]{
\newcommand{#1}[2][#2]{
  ##1  ##2 \ifnextchar\bgroup{#1[#3] 
  }{ #4 } 
}
}

\newlistcommand{\ket}{\left|}{\middle\rangle\!\middle|}{\right\rangle}
\newlistcommand{\bra}{\left\langle}{\middle|\!\middle\langle}{\right|}

\newtextfunction{\sin}{sin}
\newtextfunction{\cos}{cos}
\newtextfunction{\tan}{tan}
\newtextfunction{\arcsin}{arcsin}
\newtextfunction{\arccos}{arccos}
\newtextfunction{\arctan}{arctan}
\newtextfunction{\asin}{asin}
\newtextfunction{\acos}{acos}
\newtextfunction{\atan}{atan}

\newtextfunction{\log}{log}
\newtextfunction{\loglog}{log\ log}

\newtextfunction{\diag}{diag}
\newtextfunction{\range}{range}
\newtextfunction{\imag}{im}
\newtextfunction{\rank}{rank}
\newtextfunction{\spec}{Spec}
\newtextfunction{\ker}{ker}
\newtextfunction{\spa}{span} 
\newtextfunction{\nnz}{nnz}
\newtextfunction{\sign}{sign}
\newtextfunction{\dim}{dim}
\newtextfunction{\dom}{dom}
\newtextfunction{\tr}{Tr}
\newtextfunction{\deg}{deg}
\newtextfunction{\arg}{Arg}
\newtextfunction{\argmax}{argmax}

\newfunction{\bigO}{\mathcal{O}}
\newfunction{\littleO}{o}

\newfunction{\bigOmega}{\Omega}
\newfunction{\bigTheta}{\Theta}
\newfunction{\bigOt}{\widetilde{\mathcal{O}}}
\newfunction{\bigOmegat}{\widetilde{\Omega}}
\newfunction{\bigThetat}{\widetilde{\Theta}}
\newtextfunction{\polylog}{polylog}
\newtextfunction{\poly}{poly}

\newcommand{\zo}{\{0,1\}}

\title{A Framework for Distributed Quantum Queries in the CONGEST Model}
\author{Joran van Apeldoorn\thanks{QuSoft \& IViR, University of Amsterdam. Supported by NWO/OCW, as part of QSC (024.003.037)} \and Tijn de Vos\thanks{University of Salzburg, supported by the Austrian Science Fund (FWF): P 32863-N. This project has received funding from the European Research Council (ERC) under the European Union's Horizon 2020 research and innovation programme (grant agreement No 947702).}}
\usepackage{ifdraft}
\ifdraft{
 \newcommand{\authnote}[3]{{\color{#3} {\bf  #1:} #2}}
 }{
 \newcommand{\authnote}[3]{}
 }

\begin{document}
\maketitle
\begin{abstract}
The Quantum CONGEST model is a variant of the CONGEST model, where messages consist of  $\bigO(\log(n))$ qubits. In this paper, we give a general framework for implementing quantum query algorithms efficiently in a Quantum CONGEST network, using the concept of parallel-query quantum algorithms. 

We apply our framework for distributed quantum queries in two settings: problems where data is distributed over the network, and graph theoretical problems where the network defines the input.
The first setting is slightly unusual in CONGEST but here our results follow almost directly from the quantum query setting. The second setting is more traditional for the CONGEST model but here our framework requires also some classical CONGEST steps to apply. 

In the setting with distributed data, we show how a network can pick one of $k$ dates for a meeting such that a maximum number of nodes is available, using $\bigOt(\sqrt{kD}+D)$ rounds, with $D$ the network diameter. The classical complexity is linear in $k$. We also give an efficient algorithm for element distinctness: if all nodes together holds a list of $k$ numbers, we show that the nodes can determine whether there are any duplicates in $\bigOt(k^{2/3}D^{1/3}+D)$ rounds. Classically this problem requires $\bigOmegat(k+D)$ rounds. We also generalize the protocol for the distributed Deutsch-Jozsa problem from the two-party setting considered by Buhrman, Cleve, and Wigderson~\cite{BCW98} to general networks. This gives a novel separation between exact classical and exact quantum protocols in the CONGEST model.

In the setting where the input is the network structure itself, our framework almost directly allows us to recover the  $\bigO(\sqrt{nD})$ round diameter computation algorithm of Le Gall and Magniez~\cite{le2018sublinear}.
We extend this approach to also compute the radius in the same number of rounds, and to give an $\eps$-additive approximation of the average eccentricity in $\bigOt(D^{3/2}/\eps)$ rounds.

Finally, we give the first quantum speedups over classical CONGEST for the problems of cycle detection and girth computation. We detect whether a graph has a cycle of length at most $k$ in $\bigO(k+(kn)^{1/2-1/\Theta(k)})$ rounds. For girth computation, we give an $\bigOt(g+(gn)^{1/2-1/\Theta(g)})$ round algorithm for graphs with girth $g$, beating the classical $\bigOmega(\sqrt{n})$ round lower bound by Frischknecht, Holzer, and Wattenhofer~\cite{FHW12}.
\end{abstract}

\vspace{1em}
\noindent\textbf{Acknowledgements:} We would like to thank Sebastian Forster for the fruitful discussions on the topic and we would like to thank the anonymous reviewer who pointed out that we could remove the dependence on the diameter in our results for cycle finding and girth computation. 

\section{Introduction}

In distributed computing, we consider a network of processors that communicate in synchronized rounds to perform a computation. Up until recently, only classical processors and communication were considered. However, as quantum computing and communication become ever more promising fields, interest has been sparked in their application to distributed computing. Quantum computing and communication provide a stronger, and in a way more physically accurate, model of computation, allowing for speedups in many settings. 

For distributed computing such an advantage was first shown for the case for exact leader election \cite{tani2005exact}, i.e., electing a leader with zero error probability, a problem that is not possible to solve using randomized communication but has a polynomial time algorithm in the quantum setting. Although promising, later it was shown that when there is no bandwidth constraint, as in the \emph{LOCAL} model, many fundamental problems do not allow for a speedup \cite{GKM09}. When there \emph{is} a bandwidth constrained however, quantum communication can improve round complexities significantly. This model is called the \emph{Quantum CONGEST model}, and is the model considered in this paper.

\paragraph{The model.} The Quantum CONGEST model consists of a network of processors, which communicate in synchronous rounds. In each round, a processor can send information to its neighbors over a non-faulty link with limited bandwidth. We model the network of processors by a graph $G=(V,E)$, where we identify the processors with the nodes and the communication links with the edges. We write $n=|V|$, $m=|E|$, and $D$ for the diameter of the graph. Initially, processors do not know the complete network topology, but only who their neighbors are. Each node has a unique identifier of size $\bigO(\log(n))$, initially only known by the node itself and its neighbors. Computation in this model is done in rounds. At the start of each rounds, each node can send a different messages to each of its neighbors, and receive a message from each of them. The messages are of size at most $\bigO(\log(n))$ qubits, allowing identifiers to be communicated in a single round. Before the next round, each node can perform (unlimited) internal quantum computation. We measure the efficiency of an algorithm by the number of rounds. At the start of the computation the nodes do \emph{not} share any entanglement. 

\paragraph{Prior work.} The first positive results in the Quantum CONGEST model are in computing the diameter in $\bigO(\sqrt{nD})$ rounds \cite{le2018sublinear} and finding triangles in $\bigOt(n^{1/4})$ rounds \cite{izumi2019quantumtriangle}\footnote{In the abstract and introduction we use $\bigOt(\cdot)$ notation to hide poly-logarithmic factors in all relevant parameters, not only those appearing in the $\bigOt(\cdot)$.}. On the negative side, Elkin, Klauck, Nanongkai, and Pandurangan~\cite{elkin2014can} showed that a number of fundamental problems, such as minimum spanning tree, minimum cut, and shortest paths, can not be solved any faster in Quantum CONGEST than in its classical counterpart. We also note that Le Gall, Nishimura, and Rosmanis~\cite{le2019quantum} showed a separation of $O(1)$ versus $O(n)$ in the LOCAL model for a sampling problem.

\paragraph{Our contributions.} In this paper, we give a general framework for implementing quantum query algorithms efficiently in a distributed network. This allows us to take any of the many positive results in quantum query complexity and apply them to problems in the Quantum CONGEST model. 
While classical algorithms often have to query the entire input, many quantum algorithms require a sub-linear number of queries, making query complexity an interesting metric in the quantum setting. For example, one of the most famous examples of a quantum improvement comes from Grover's search algorithm, which searches an unordered input of size $N$ using only $\bigO(\sqrt{N})$ queries to that input.  An important ingredient in these speedups is the ability to query the input in \emph{superposition}.

In a CONGEST setting were a designated leader has been chosen, we can ask this leader to run some query algorithm that would solve the problem we are interested in. As the leader might not have all the information required to solve the problem (otherwise the problem would be trivial to solve in the CONGEST model!), each query requires help from the network to implement. We give a general strategy for implementing such queries efficiently for a quantum query algorithm, where the queries can be in a superposition over different values. A straight forward implementation of this technique would lead to most nodes sitting idle most of the time, as they wait for some other part of the network to finish the query. We circumvent this by turning to parallel-query quantum algorithms, where queries need to be made in simultaneous batches. While  parallel-query quantum algorithms are a natural generalization of sequential quantum query algorithms, few results are known. In \autoref{sc:parallel} we give an overview of the few results that we are aware of, as well as giving alternative proofs and some small improvements for some of these results.

We apply our framework for distributed quantum queries in two settings: problems where data is distributed over the network, and graph theoretical problems where the network defines the input.
The first setting is slightly unusual in CONGEST but here our results follow almost directly from the quantum query setting. The second setting is more traditional for the CONGEST model but here our framework requires also some classical CONGEST steps to apply. 

Concerning distributed data problems, we show the following bounds on the round complexity:
\begin{itemize}
    \item We show that $\bigOt(\sqrt{kD}+D)$ rounds are sufficient for the meeting scheduling problem. Here each node has a private ``calendar'' and the goal is to pick one out of $k$ possible days at which most nodes are ``available''. 
    \item We show that $\bigOt(k^{2/3}D^{1/3}+D)$ rounds are sufficient to solve a distributed version of element distinctness. Here the nodes together hold a list of $k$ elements, and the goal is to find two elements in the list that are the same (or conclude that no such pairs exist). 
    \item Buhrman, Cleve, and Wigderson~\cite{BCW98} previously extended the Deutsch-Jozsa problem~\cite{DJ92} to a two-party communication setting. We further extend this to arbitrary networks, leading to a novel proof of an exponential separation between exact (i.e., with zero error probability) classical CONGEST algorithms, and exact Quantum CONGEST algorithms. This however leaves the interesting open question of such a separation also exists for a natural problem and for algorithms with a small error probability.
\end{itemize}
To prove the corresponding lower bounds for the these problems, we use reductions from two-party communication complexity. In particular, for the first two we make use of the intersection problem (also known as the disjointness problem), where two players receive $x^{(A)},x^{(B)}\in\zo^k$ respectively and they have to determine whether there exists an index $i$ such that $x^{(A)}=x^{(B)}=1$. This takes $\Omega(k)$ rounds of communication classically~\cite{KS87,Razborov90}, where each round consists of communicating $\bigO(1)$ bits. For the third problem, we use the two-player lower bound of Deutsch-Jozsa, which is also $\Omega(k)$ rounds. This leads to a lower bound of $\bigOmega(k/\log(n)+D)$ rounds for all three problems in the classical CONGEST model. In particular, for Deutsch-Jozsa this means we obtain an exponential speedup with respect to the classical CONGEST model. Moreover, as there is a trivial lower bound of $\bigOmega(D)$ rounds in the quantum CONGEST model, our solution is optimal up to log-factors. For meeting scheduling and element distinctness, an exponential speedup is not possible. In fact, we show both problems require $\bigOmega(\sqrt[3]{kD^2}+\sqrt{k})$ rounds in the quantum CONGEST model. Note that here we only consider $k\geq D$, since the complexity of $k< D$ is $\Theta(D)$, where the upper bound is achieved by simply streaming the entire input to the other party.

For the graph theoretical problems, we show the following bounds on the round complexity:
\begin{itemize}
    \item We directly recover the $\bigO(\sqrt{nD})$ round algorithm for diameter computation by Le Gall and Magniez~\cite{le2018sublinear}. We also extent these idea to radius computation and approximate average eccentricity. For the later problem we give an $\bigOt(D^{3/2}/\eps)$ round algorithm that find an $\eps$-additive approximation of the average eccentricity.
    \item We show that we can detect whether a graph has a cycle of length at most $k$ in $\bigO(k + (kn)^{\frac{1}{2}-\frac{1}{4\ceil{k/2}+2}})$ rounds in the Quantum CONGEST model. 
    \item  We also show that the girth $g$ of a graph can be computed in \[
    \bigOt(g+(gn)^{\frac{1}{2}-\frac{1}{\Theta(g)}})
    \]
    rounds, where any classical algorithm would require $\bigOmega(\sqrt{n})$ rounds. 
\end{itemize}

We finish the paper with a short discussion about implementing quantum algorithms that do not use queries to a standard oracle. In particular, we show how to implement amplitude amplification, phase estimation, and amplitude estimation in the quantum CONGEST model. All of these algorithmic building blocks work with more general black-box (quantum) subroutines that do not need to take the form of a standard input oracle.

\paragraph{Notation and conventions.}
Following conventions from quantum computing, we prove our results with success probability at least $2/3$. In our algorithms, there will always be some central leader that can combine the results of multiple independent runs to boost this to a success probability of $1-n^{-c}$ at the cost of an extra $\log(n)$-factor.  

Throughout this paper, we assume that the reader is familiar with the basic concepts of quantum computing, for a good overview see for example~\cite{dewolf:notes}.

 \section{Parallel Quantum Queries}\label{sc:parallel}
 Before we consider distributed computing, we look at quantum query complexity. The standard model of quantum query complexity is fully adaptive, all queries happen in sequence. For many quantum query upper bounds this adaptiveness is key, for example, the speedup achieved by Grover's algorithm disappears when all queries have to be performed in parallel~\cite{zalka:optimalgrover}. More generally, non-adaptive quantum query algorithms are very limited in the speedup they can achieve~\cite{montanaro:nonadaptive}.
 
 More recently, attention has been given to parallel quantum queries~\cite{jeffery:parallelquantum,zalka:optimalgrover,grover:parallel,giurgicatiron:ampest}, a natural interpolation between the non-adaptive and fully-adaptive setting. Here we make queries in batches of $p$ simultaneous queries. 

 \begin{definition}[Parallel-query Quantum Algorithm]  \label{def:batched-query}
Let $O$ be an input oracle for a query problem. A $(b,p)$-parallel-query quantum algorithm is a quantum query algorithm that contains $b$ uses of $O^{\otimes p}$, i.e., it contains $b$ batches of $p$ parallel queries.
\end{definition}
 For completeness we reprove a few results about parallel-query quantum algorithms from the literature, with a few small improvements.

 \begin{lemma}[Parallel Grover's Search]\label{lm:batched-grover} 
    Let $x\in \zo^k$ be an input string with $|x| = t$. Then there exists a $\left(\bigO(\ceil{\sqrt{k/(tp)}}),p\right)$-parallel-query quantum algorithm that finds an index $i$ such that $x_i=1$, or outputs that no such index exists. There also exists a $(\bigO(\sqrt{kt/p}+t),p)$-parallel-query quantum algorithm that finds all such indices. Both algorithms succeed with probability at least $2/3$.
\end{lemma}
\begin{proof}
  A simple approach would be as follows: split the input string into $p$ parts of size $k/p$ and apply Grover's algorithm to each part. However, to ensure with constant probability all runs of Grover succeed, each run had to fail with probability at most $1/(3p)$, introducing an extra $\log(p)$ factor. This was the approach taken in previous works~\cite{zalka:optimalgrover,grover:parallel}.

  To get around this we search over subsets of size $p$ of the input instead. We first focus on the expected number of queries for finding one element if there are exactly $t$ marked elements (but we do not know this value). If $p\geq k/t$ then a random subset of $p$ elements contains a marked element and hence the expected number of queries to find one is constant. Therefore we assume $p< k/t$.
  
  There are $\binom{k}{p}$ subsets of the input, out of which $\binom{k-t}{p}$ \emph{do not} contain a marked element. Hence, using that $\binom{a}{b} = \bigTheta(\frac{a^b}{b!})$, we know that Grover's algorithm requires an expected number of queries equal to~\cite{boyer:search} 
  \begin{align*}
    \bigO(\sqrt{\frac{\binom{k}{p}}{\binom{k}{p} - \binom{k-t}{p}}}) &=   \bigO(\sqrt{\frac{k^p}{k^p - (k-t)^p}})\\
    &=   \bigO(\sqrt{\frac{1}{1 - \trm{1-\frac{t}{k}}^p}})\\
    &=   \bigO(\sqrt{\frac{1}{\frac{pt}{k}}})\\
    &=   \bigO(\sqrt{\frac{k}{pt}}),
  \end{align*}
  where we used that $1-(1-f)^p \geq fp$ if $p < 1/f$.
  Hence the expected number of parallel queries to find the first element is $\bigO(\ceil{\sqrt{\frac{k}{pt}} })$.

  Summing this over all $t$ elements, with decreasing values for $t$ we get
  \[
  \bigO(\sum_{\tau=1}^t \ceil{\sqrt{\frac{k}{p\tau}}} ) =  \bigO(\sqrt{\frac{kt}{p}} + t)
  \]
  for the expected number of parallel queries if $t$ elements were marked. As the algorithm is agnostic to the value of $t$, we may use our given upper bound to get an upper bound on the expected number of queries in our setting.
  
  Stopping any of the algorithms after $3$ times their expected number and applying Markov's inequality gives use the same worst-case bound, with error probability $\leq 1/3$.
\end{proof}

A well known extension of Grover's search algorithm is the minimum finding algorithm by D\"urr and H{\o}yer~\cite{durr1996QMinimumFinding}. This can be parallelized in a similar manner, or by directly using the expected number of queries bound from the proof above in their proof.  Hence, we will not give an explicit proof here.
\begin{lemma}[Parallel minimum finding]\label{lm:batched-minfinding} 
    Let $x\in [N]^k$ be an input string. There exists a $\left(\bigO(\ceil{\sqrt{k/p}}),p\right)$-parallel-query quantum algorithm that finds an index $i$ such that, with probability at least $2/3$, $x_i=\min_j x_j$. Equivalently, the algorithm can be modified output the maximum as well.
    
    If there are at least $\ell$ elements that attain the minimum (or equivalently maximum) than there eqists a $\left(\bigO(\ceil{\sqrt{\frac{k}{\ell p}}}),p\right)$-parallel-query quantum algorithm for the problem.
\end{lemma}

Next, we consider the problem of element distinctness.
\begin{definition}[Element distinctness] The \emph{element distinctness problem} is the problem of deciding whether an input string $x\in [N]^k$ contains the same element twice, and if so, finding the matching elements.
\end{definition}
Recently, Jeffery, Magniez, and de Wolf~\cite{jeffery:parallelquantum} gave an optimal $\left(\bigO(\ceil{\left(\frac{k}{p}\right)^{2/3}}),p\right)$-parallel-query quantum algorithm for this problem. They modified the quantum walk algorithm by Ambainis~\cite{Ambainis07} to walk over multiple Johnson graphs at the same time. Below we reprove their result with a slightly different approach, where $p$ classical random walk steps are taken to create a single quantum walk step. We will assume knowledge of the quantum walk framework used and the non-parallel quantum algorithm for element distinctness, for a good overview see for example Section~8 of~\cite{dewolf:notes}.
\begin{lemma}\label{lm:element-distinctness}
  Let $x\in [N]^k$ be a string such that at least one disjoint pairs of indices $(i,j)$ exists with $i\neq j$ and $x_i = x_j$. There exists a $\left(\bigO(\ceil{\frac{k^{2/3}}{p^{2/3}}}),p\right)$-parallel-query quantum algorithm that, with probability at least $2/3$, outputs one such pair. 
\end{lemma}
\begin{proof}
  We will use the same method as Ambainis, but rebalance the subsets over which the walk runs slightly, and we perform multiple classical steps in parallel.
  
We first consider the probability that a subset of the indices of size $z$ contains at least one of the pairs.  The probability that a specific pair is fully in the subset is larger than $\eps=\frac{z^2}{k^2}$. 
Hence, if $p \geq \frac{k}{8}$, then one parallel query suffices to fully query a random subset of size $z = \frac{k}{8}$ and find a pair with constant probability, so we are done. Therefore, we assume $p<\frac{k}{8}$ for the rest of the proof.
  
   For this case, we use a quantum random walk. Consider the Johnson graph $J(k,z)$ (for some $z$ to be determined later), which has vertices labeled with $z$-sized subsets of $[k]$, and edges between subsets that differ by one replacement. This graph has spectral gap $\delta = \Omega(\frac{1}{z})$ (for $z \leq k/2$) \cite{brouwer:spectra}. Consider the $p$th power of this graph, i.e., the graph corresponding to $p$ steps of a random walk. Its spectral gap is at least $1-(1-\delta)^p = \bigOmega(p\delta) = \bigOmega(p/z)$ (as long as $p<\frac{1}{\delta} = z$).  During the walk we will keep track of the indices in a subset and their values. Setup, denoted by $S$, requires $\bigO(z/p)$ parallel queries. The update to the indices in each step can be performed with no queries, after which $p$ queries in parallel suffice to update their corresponding values. Hence the total update, denoted by $U$, can be done using a single parallel query.
  We will call a subset marked if it contains at least one pair of indices that have the same value. As we keep track of the values corresponding to indices, checking whether a subset is marked, denoted by $C$, requires no queries.

The number of parallel queries required then is 
  \[
  S+\frac{1}{\sqrt{\eps}} \trm{C+\frac{1}{\sqrt{\delta}}U} = \bigO(\frac{z}{p} + \frac{k}{z} \frac{\sqrt{z}}{\sqrt{p}} \cdot 1) = \bigO(\frac{z}{p}+\frac{k}{\sqrt{zp}}),
  \]
  which is minimal if $z=k^{2/3}p^{1/3}$, resulting in a quantum parallel-query complexity of $\bigO(\ceil{\frac{k^{2/3}}{p^{2/3}}})$.
  
  Finally we check that the choice of parameters satisfies the two mentioned requirements: $p< z$ and $z\leq k/2$. As $p < \frac{k}{8}$, we indeed have $p< (k/8)^{2/3}p^{1/3}<z$. And by the same reasoning, we have $z=k^{2/3}p^{1/3}<k^{2/3}(k/8)^{1/3}=k/2$.
\end{proof}
 Finally, we consider the problem of amplitude estimation. Recently Giurgica-Tiron et al.~\cite{giurgicatiron:ampest} gave two quantum algorithms for parallel-query amplitude estimation that achieve the optimal trade-off between number of queries and depth.  Below we give a much simplified proof for the more general setting of estimating an expectation value (instead of just a probability). Our result matches those of Giurgica-Tiron et al.\ up to logarithmic factors for the case of probability estimation.
 \begin{lemma}\label{lm:parallel-expectation}
     Let $U_X$ be a unitary matrix that, starting from $\ket{0}$, creates a quantum state such that a measurement would sample from a random variable $X$ with expectation value $\mu$ and variance $\sigma^2$. There exists a $\left(\bigO(\ceil{\frac{\sigma}{\sqrt{p}\eps} \cdot \log^{3/2}(\frac{\sigma}{\sqrt{p}\eps}) \log\log(\frac{\sigma}{\sqrt{p}\eps})}),p\right)$-parallel-query quantum algorithm that uses $U_X$ and its conjugate transpose $U_X^{\dagger}$ as input oracles and returns an $\eps$-additive-estimate of $\mathbb{E}[X]$ with success probability at least $2/3$.  
 \end{lemma}
 \begin{proof}
    The lemma for the $p=1$ case is Theorem~5 in~\cite{montanaro:montecarlo}.

    Now consider the random variable $Y$ corresponding to the average of $p$ samples from $X$. Clearly a unitary $U_Y$ can be constructed using $p$ applications of $U$. Furthermore, the variance of $Y$ is $\frac{\sigma^2}{n}$. Applying the lemma for the $p=1$ case with $Y$ and $\sigma' = \sigma/\sqrt{p}$ therefore gives the result.
 \end{proof}

\section{Making Parallel Queries in the CONGEST Model}\label{sc:mainthm}
In this section, we present our main contribution: we show how to change a register of $q$ qubits located at one node, into a state that is distributed through the system and vice versa. For the precise statement, see \autoref{lm:The-Van-Apeldoorn-De-Vos-Lemma}. We then continue to show how this can be used to evaluate functions were the input is distributed over the nodes, see \autoref{thm:mainthm}. This is done by selecting a leader who runs the central quantum algorithm, where each query has to be made to the network.

\begin{lemma}\label{lm:The-Van-Apeldoorn-De-Vos-Lemma}
    Let $Q = 2^q$. In the Quantum CONGEST model, assume there is a designated leader. If the leader holds a state $\sum_i^Q \alpha_i\ket{i}$, then there is a $\bigO\left(D + \frac{q}{\log(n)}\right)$ round algorithm that transforms the state into $\sum_i^Q \alpha_i\ket{i}^{\otimes n}$, where each node holds one of the registers.
    
    The reverse is also possible, if the nodes together hold $\sum_i^Q \alpha_i\ket{i}^{\otimes n}$, then there is a  $\bigO\left(D + \frac{q}{\log(n)}\right)$ round algorithm that transforms the state into $\sum_i^Q \alpha_i\ket{i}$, situated at the leader.
\end{lemma}
\begin{proof}
    First of all, we construct a BFS tree from the leader in $\bigO(D)$ rounds\footnote{This is done by the folklore algorithm, where, starting with the leader, each node declares itself scanned in round $i$ if one of their neighbors declared itself scanned in round $i-1$. In case multiple neighbors do so, pick any to be the parent on the tree.}. This tree will be used for the communication, where we will send information down to all the nodes from the leader. We perform the analysis on a single basis state $\ket{i}$, the result then follows by linearity. 
    
    At the start of the algorithm the leader creates a new register for each of its children. The leader then applies CNOT gates from its starting state $\ket{i}$ to all of the registers, creating the state $\ket{i}^{\otimes k+1}$, where $k$ is the number of children. It can then send each of its children their register using $\ceil{q/\log(n)}$ rounds. This is then repeated for each of the at most $D$ layers of the tree. 
    
    Naively this would result in an $\bigO(D\ceil{q/\log(n)})$ round complexity. However, note that a node does not have to wait until it received the full register $\ket{i}$ before it can start sending to its children. Each $\log(n)$ message can be passed on in the next round. This leads to a $\bigO(D+q/\log(n))$ round complexity.
    For the reverse result, we run the same algorithm in reverse. 
\end{proof}

Note that leader election can be done in $\bigO(D)$ in the CONGEST model, so if no leader is provided, we can for example take the node with the largest identifier. 

We now show that, given a parallel-query quantum algorithm, we can transform this into a Quantum CONGEST algorithm, where the queries have to be made with the help of the network. We show how to make these queries using the previous lemma, and how the parallel-query upper bound of the quantum algorithm leads to a upper bound on the number of rounds in the CONGEST model. 

This method is somewhat similar to the more limited method used by Le Gall, and Magniez~\cite{le2018sublinear}. They show that the queries in the non-parallel version of Grover's search algorithm and maximum finding can be implemented in the CONGEST model. Using the non-parallel version of query algorithms can be sub-optimal as for high diameter graphs most nodes will be sitting idle most of the time as the query makes its way through the network. By performing multiple queries at once this can be circumvented. 

\begin{theorem}\label{thm:mainthm}
  Let $f \colon A^{k\times n} \rightarrow R$ be defined as $F\left(\bigoplus_{v\in V} x^{(v)}\right)$ for some $F:A^k\rightarrow R$ and some operation $\oplus \colon A\times A \rightarrow A$ that is performed element wise. Let $q = \ceil{\log(|A|)}$. If $(A,\oplus)$ forms a commutative semigroup, and $F$ has parallel-query quantum complexity $(b,p)$, then $f$ can be evaluated in the Quantum CONGEST model in
  \[
    \bigO\trm{D + b\cdot \trm{(D+p) \ceil{\frac{q}{\log(n)}} + p\ceil{\frac{\log(k)}{\log(n)}}}}  
  \]
  rounds.
\end{theorem}
\begin{proof}
    The idea is similar to \autoref{lm:The-Van-Apeldoorn-De-Vos-Lemma}, start with leader election and BFS, in $\bigO(D)$ rounds. After this, the leader runs the query algorithm and for each batch of $p$ parallel queries uses the result from \autoref{lm:The-Van-Apeldoorn-De-Vos-Lemma} on $\otimes_{i=1}^p \ket{j_i}$ to distribute the indices $j_1, \dots, j_p \in[k]$ through the network, costing $\bigO\trm{D+ p\ceil{\frac{\log(k)}{\log(n)}}}$ for a batch. 
    
    Now for all $i\in [p]$, all leaf nodes $w$ send the query results $x^{(w)}_{j_i}$ to their parent, who computes $\bigoplus_{v} x^{(v)}_{j_i}$, with $v$ ranging over itself and its children. It then sends the $x^{(w)}_{j_i}$ back to the children, who uncompute it. This continues through the tree. Note that here we can not fully stream the results as before, as a node needs the full values of its children before it can compute the $\oplus$ operation. Sending one result up one level requires $\bigO(\ceil{q/\log(n)})$ rounds. As soon as the leaves are done with the first query value they can  start with the second query value, and so on. In total this will require $D+p$ sets of $\bigO(\ceil{q/\log(n)})$ rounds to compute all queries at the leader. 
    
    After this, the simulated query ends with the reverse of the distribution of the indices $j_1,\dots,j_p$. The stated complexity then follows from the complexity of the parallel query algorithm.
\end{proof}
This theorem can directly be applied if the nodes hold their data in memory. However, if we want to apply the theorem to graph theoretic problems (as in \autoref{sc:graph-problems}), then this might not be the case. In that setting the nodes will often have to perform a short CONGEST algorithm to compute the values. For example, for the computation of the diameter we will query the eccentricity of a node. To compute this eccentricity we first computed BFS from the node. Computing all these values before applying \autoref{thm:mainthm} and saving them in memory takes too many rounds. However, the following corollary allows us to compute value on-the-fly. 

\begin{corollary}\label{cor:thm-with-query-time}
    Suppose we are in the setting of \autoref{thm:mainthm}, but the values $x^{(v)}_j$, are not known in advance. Rather, we can compute the $p$ values of a batch in $\alpha(p)$ rounds. Then we can evaluate $f$ in the Quantum CONGEST model in 
  \[
    \bigO\trm{D + b\cdot \trm{(D+p) \ceil{\frac{q}{\log(n)}} + p\ceil{\frac{\log(k)}{\log(n)}}+\alpha(p)}}  
  \]
  rounds. 
\end{corollary}
\begin{proof}
    This follows directly from \autoref{thm:mainthm}, we apply the same algorithm, but for each batch of queries, we first compute the $p$ values we want to query in $\alpha(p)$ rounds.
\end{proof}

\section{Applications to Distributed Data Problems}\label{sc:data_problems}
In this section we give a few applications of \autoref{thm:mainthm}. In particular, we consider the setting were the nodes already hold the data $x^{(v)}$ in memory. This setting is less common for the CONGEST model but suits our framework very well. 
\subsection{Meeting Scheduling}
Informally, the meeting scheduling problem is as follows: suppose there are $n$ participants and $k$ time slots. Given that each participant knows which slots they are available, determine the time slot with the maximum number of available participants. 

\begin{lemma}[Meeting Scheduling]
Given $x^{(v)}\in \{0,1\}^k$, for each $v\in V$. In the Quantum CONGEST model we can compute $\argmax_{i\in[k]} \sum_{v\in V} x^{(v)}_i$ with success probability at least $2/3$ in 
\[
    \bigO((\sqrt{kD}+D)\ceil{\frac{\log(k)}{\log(n)}})
\]
rounds of communication.
\end{lemma}
\begin{proof}
    We use the parallel-query maximum finding algorithm of \autoref{lm:batched-minfinding} with $p=D$, giving us $b=\bigO(\ceil{\sqrt{k/D}})$. Now we apply \autoref{thm:mainthm} for $A = [n]$ and $\oplus = +$. This gives an 
    \begin{align*}
        \bigO\left( D+b\left((D+p)\ceil{\frac{q}{\log(n)}}+p\ceil{\frac{\log(k)}{\log(n)}}\right)\right) &= \bigO\left( D+\ceil{\sqrt{k/D}}\left(2D+D\ceil{\frac{\log(k)}{\log(n)}}\right)\right) \\
        &= \bigO\left((\sqrt{kD}+D)\ceil{\frac{\log(k)}{\log(n)}}\right)
    \end{align*}
    round algorithm. 
\end{proof}
Note that this can be generalized to other domains $A$ and non-zero-one inputs, at the cost of an extra $q = \log(|A|)$ factor.

We can show a separation between the quantum and classical CONGEST models by proving a lower bound in the classical case. 
\begin{lemma}\label{lm:LB-meeting}
Solving the meeting scheduling problem with success probability $\geq 2/3$ requires $\bigOmega(k/\log(n)+D)$ rounds in the classical CONGEST model and $\bigOmega(\sqrt[3]{kD^2}+\sqrt{k})$ rounds in the quantum CONGEST model.
\end{lemma}
\begin{proof}
    We will reduce two-party disjointness to the meeting scheduling problem. For the two party disjointness problem, we have parties $A$ and $B$, with inputs $x^{(A)} \in \zo^k$ and $x^{(B)} \in \zo^k$ respectively. The question is if $\max_j(x^{(v_A)}_j+x^{(v_B)}_j)=2$. This takes takes $\bigOmega(k)$ rounds of communication with messages of size $\bigO(1)$ \cite{KS87,Razborov90}, even if we just require to be correct with constant probability. Now suppose there is a $\littleO(k/\log(n))$-round CONGEST algorithm for the meeting scheduling problem, then this gives rise to a $\littleO(k)$-round algorithm for the two party disjointness problem in the following manner. Consider the graph consisting of two nodes $v_A$ and $v_B$ distance $D$ apart. Suppose we have $x^{(v)}=0^k$ for $v\neq v_A,v_B$ and $x^{(v_A)}=x^{(A)}$ and $x^{(v_B)}=x^{(B)}$. Now the meeting scheduling algorithm will tell us $\argmax_{i\in[k]} \sum_{v\in V} x^{(v)}_i$, in an additional round of communication, this gives us whether $\max_j(x^{(v_A)}_j+x^{(v_B)}_j)=2$, hence solving the instance of two-party disjointness. This CONGEST algorithm becomes a two-party algorithm, by letting $A$ simulate $v_A$ and letting $B$ simulate all other nodes. The communication between $A$ and $B$ is now the communication between $v_A$ and its neighbor, which is $\littleO(k/\log(n))$ rounds of communication with messages of size $\bigO(\log(n))$. By reducing message size to $\bigO(1)$, we get a $\littleO(k)$-round algorithm for two-party disjointness, a contradiction. 
    
    Similarly, we get a lower bound of $\bigOmega(D)$, since $v_A$ and $v_B$ as defined above are $D$ apart, we obtain a lower bound of $\bigOmega(k/\log(n)+D)$ in the classical CONGEST model.
    
    For the quantum lower bound, we use the same reduction, together with the fact that set disjointness on a line requires $\bigOmega(\sqrt[3]{kD^2}+\sqrt{k})$ rounds in the quantum CONGEST model~\cite{magniez2020quantum}.
\end{proof}
Note that there exists a trivial $\bigO(k/\log(n)+D)$ classical CONGEST algorithm to solve the problem where all nodes send all their values to a leader through the BFS tree. This can be seen as coming from a trivial classical parallel-query algorithm: query all values in one batch of size $p=k$. 

\subsection{Element Distinctness}
Next, we consider the problem of element distinctness. When generalizing the query version of this problem to the CONGEST model, two natural versions arise. One option is that each node holds a vector $x^{(v)}\in [N]^k$, and the vector we want to perform element distinctness on is the sum of these: $x=\sum_{v\in V}x^{(v)}$. Alternatively, we may have that each node holds an entry of the complete vector, i.e., we have $x^{(v)}\in [N]$ and look at $x=(x^{(v)})_{v\in V}$. The problems are solved very similarly, both by using the parallel-query algorithm for element distinctness. 

\begin{lemma}[Element Distinctness in Distributed Vector]\label{lem:elldistvec}
Suppose each node $v$ has a vector $x^{(v)}\in [N]^k$, element distinctness for $x=\sum_v x^{(v)}$ can be solved with probability at least $2/3$ in \[\bigO\left( (k^{2/3}D^{1/3}+D)\left(\ceil{\frac{\log(N)}{\log(n)}}+\ceil{\frac{\log(k)}{\log(n)}}\right)\right)\] rounds in Quantum CONGEST. 
\end{lemma}
\begin{proof}
    We use the parallel-query algorithm of \autoref{lm:element-distinctness} with $p=D$, this gives us $b=\bigO(\ceil{k^{2/3}/D^{2/3}})$. Now we apply \autoref{thm:mainthm} with $A=[Nn]$ and $\oplus = +$, which gives us \begin{align*}
        \bigO\left( D+b\left((D+p)\ceil{\frac{q}{\log(n)}}+p\ceil{\frac{\log(k)}{\log(n)}}\right)\right) &= \bigO\left( D+\ceil{\frac{k^{2/3}}{D^{2/3}}}\left(2D\ceil{\frac{\log(N)}{\log(n)}}+D\ceil{\frac{\log(k)}{\log(n)}}\right)\right) \\
        &= \bigO\left((k^{2/3}D^{1/3}+D)\left(\ceil{\frac{\log(N)}{\log(n)}}+\ceil{\frac{\log(k)}{\log(n)}}\right)\right)
    \end{align*}
    rounds.
\end{proof}

In the next lemma we give a lower bound for this problem in the classical CONGEST model, which our quantum algorithm beats for large enough $k$. 

\begin{lemma}\label{lm:LB_element_distinctness_vector}
The element distinctness in distributed vector problem requires $\bigOmega(k/\log(n)+D)$ rounds in the classical CONGEST model and $\bigOmega(\sqrt[3]{kD^2}+\sqrt{k})$ rounds in the quantum CONGEST model.
\end{lemma}
\begin{proof}
We will reduce two-party disjointness to element distinction, as in \autoref{lm:LB-meeting}. Let parties $A$ and $B$ have inputs $x^{(A)} \in \zo^k$ and $x^{(B)} \in \zo^k$ respectively. The disjointness problem is if $\max_j(x^{(v_A)}_j+x^{(v_B)}_j)=2$. This takes takes $\bigOmega(k)$ rounds of communication with messages of size $\bigO(1)$ \cite{KS87,Razborov90}, even if we just require to be correct with constant probability. Now suppose there is a $\littleO(k/\log(n))$-round CONGEST algorithm for the element distinctness in distributed vector problem, then this gives rise to a $\littleO(k)$-round algorithm for the two party disjointness problem in the following manner. Again, we consider two nodes $v_A$ and $v_B$ that are $D$ hops apart. Suppose we have $x^{(v)}=0^{2k}$ for $v\neq v_A,v_B$ and
\begin{align*}
    x^{(v_A)}_i &= \begin{cases} i &\text{if } x^{(A)}_i=1 \text{ and }i \leq k;\\
    2k+i &\text{if } x^{(A)}_i=0 \text{ and }i \leq k;\\
    0 &\text{if }i>k;\end{cases}\\
    x^{(v_B)}_i &= \begin{cases} 0 &\text{if }i \leq k;\\
    i-k &\text{if } x^{(B)}_{i-k}=1 \text{ and }i > k;\\
    3k+i &\text{if } x^{(B)}_{i-k}=0 \text{ and }i > k. \end{cases}
\end{align*}
Now note that we have a collision in $x^{(v_A)}+x^{(v_B)}$ if and only if $\max_j(x^{(v_A)}_j+x^{(v_B)}_j)=2$. Hence distinctness in distributed vector solves element disjointness. We make the CONGEST algorithm into a two-party algorithm by letting $A$ simulate $v_A$ and $B$ simulate all other nodes. As before, this gives a $\bigOmega(k/\log(n))$ lower bound. 

Moreover, since the nodes are $D$ apart, we also have a $\bigOmega(D)$ lower bound, giving us $\bigOmega(k/\log(n)+D)$ in total. 
    
For the quantum lower bound, we use the same reduction, together with the fact that set disjointness on a line requires $\bigOmega(\sqrt[3]{kD^2}+\sqrt{k})$ rounds in the quantum CONGEST model~\cite{magniez2020quantum}.
\end{proof}

We note that this also implies an algorithm for a more natural version of the element distinctness problem in CONGEST, where each node holds one value and we want to know whether these are all distinct. However, as pointed out by an anonymous reviewer, there also exists an $\bigOt(\sqrt{nD})$ round algorithm using hash functions and a (parallel) Grover's search. 

\begin{corollary}[Element Distinctness Between Nodes]
Suppose each node $v$ holds a value $x^{(v)}\in [N]$, we can solve the element distinctness problem for $x=(x^{(v)})_{v\in V}$ with probability at least $2/3$ in \[
\bigO\left((n^{2/3}D^{1/3}+D)\ceil{\frac{\log(N)}{\log(n)}}\right)\]
rounds in the Quantum CONGEST model. 
\end{corollary}
\begin{proof}
This follows directly from \autoref{lem:elldistvec} by redefining the input for each node to be a length $k=n$ vector where just one element is non-zero and equal to its actual input.
\end{proof}

Again, we give a lower bound in the classical CONGEST model, which our quantum algorithm beats for small enough $D$. 

\begin{lemma}\label{lm:LB_element_distinctness_nodes}
The element distinctness between nodes problem requires $\Omega(n/\log(n))$ rounds in the classical CONGEST model and $\bigOmega(\sqrt[3]{nD^2}+\sqrt{n})$ rounds in the quantum CONGEST model. 
\end{lemma}
\begin{proof}
We will reduce from two-party disjointness, as in \autoref{lm:LB_element_distinctness_vector}. Let $S_A$ be the indices for which the input of $v_A$ is $1$, and let $S_B$ be similarly defined for $v_B$. Consider a network that consists of two star graphs of sizes $|S_A|$ and $|S_B|$, connected with an edge between their centers. Each point of the star gets a single element of $S_A$ or $S_B$. If there would be a classical CONGEST algorithm that solves element distinctness on this network and for this input, then this would also solve the disjointness problem. As the two players could each simulate one of the stars, and would only have to communicate for the bits send over the connecting edge, at least $\bigOmega(n)$ bits would have to be send over this edge \cite{KS87,Razborov90}. Hence the element distinctness between nodes problem requires at least $\Omega(n/\log(n))$ rounds in the classical CONGEST model.

For the quantum lower bound, we use the same reduction, together with the fact that set disjointness on a line requires $\bigOmega(\sqrt[3]{nD^2}+\sqrt{n})$ rounds in the quantum CONGEST model~\cite{magniez2020quantum}.
\end{proof}

\subsection{Distributed Deutsch-Jozsa}
In this section, we consider a distributed version of the Deutsch-Jozsa problem. The original problem was first introduced by Deutsch and Jozsa~\cite{DJ92} as a computational problem, where the goal was to minimize query complexity. They showed there was an exponential zero-error quantum speedup with respect to any zero-error classical algorithm for this problem. Later, the problem was considered as a two-party communication problem by Buhrman, Cleve, and Wigderson~\cite{BCW98}. Also here, the authors showed that quantum communication allows for an exponential speedup with respect to any zero-error classical communication protocol. In this section we make the natural generalization to $n$ parties. Again, we show that this gives an exponential speedup with respect to any zero-error classical CONGEST algorithm.

\begin{problem}[Distributed Deutsch-Jozsa]
    Let $k$ be an even positive integer. Each node $v \in V$ in the network receives a string $x^{(v)}\in \zo^k$ with the promise that for $x = \bigoplus_{v\in V} x^{(v)}$ (element-wise XOR) we have either:
        \begin{enumerate}
            \item $x$ is constant, i.e., $x\in \{0^k,1^k\}$, or
            \item $x$ is balanced, i.e., $|x| = k/2$.
        \end{enumerate}
        Determine which of these cases holds.
\end{problem}

For this definition we directly get the following result.
\begin{theorem}
    In the Quantum CONGEST model, the Distributed Deutsch-Jozsa problem can be solved with probability 1 using $\bigO(D\ceil{\frac{\log(k)}{\log(n)}})$ rounds of communication. 
\end{theorem}
\begin{proof}This follows directly from an application of \autoref{thm:mainthm} to the $\bigO(1)$ query algorithm for the Deutsch-Jozsa problem~\cite{DJ92}.

\end{proof}

As said, this is an exponential speedup with respect to any classical algorithm that has to be correct with probability 1.
\begin{theorem}
    In the classical CONGEST model, the Distributed Deutsch-Jozsa problem requires at least $\bigOmega(k/\log(n)+D)$ rounds of communication to be solved with probability $1$.
\end{theorem}
\begin{proof}
    We will reduce two-party Deutsch-Jozsa to distributed Deutsch-Jozsa, similar to \autoref{lm:LB-meeting}. Consider line graph of length $D$, with non-trivial input at ends, and constant zero in between. Now this reduces to the two party complexity of Deutsch-Jozsa, which is $\bigOmega(k)$ \cite{BCW98}, using messages of size $\bigO(1)$. Of course we need at least time $D$ to get any message from $v_A$ to $v_B$, so the total lower bound becomes $\bigOmega(k/\log(n)+D)$. 
\end{proof}

Note this speedup with respect to the classical CONGEST model only holds when we want to outcome to be correct with probability $1$. When we allow for an error probability, there is a simple and fast classical algorithms: pick a small number of indices, check whether these all have value $0$ or all have value $1$, and if so output that $x$ is constant, otherwise output that $x$ is balanced. It remains an interesting open problem to find a natural problem\footnote{In the two player communication setting such results are known, and these could be directly applied to more complicated networks. However, this would give a rather unnatural problem where two nodes in the network receive an input and the network is simply used in place of a simple communication link. } for the CONGEST model where quantum communication gives an exponential speedup over bounded-error classical algorithms as well.

\section{Applications to Graph Problems}\label{sc:graph-problems}
\autoref{cor:thm-with-query-time} show that our framework also works for problems where the query result for each node requires some computation. In this section we consider a few such problems, in particular, we consider graph problems where a classical CONGEST algorithm is used to answer each query.

\subsection{Computing the Diameter, Radius, and Average Eccentricity}
First of all, let us define the eccentricity. 
\begin{definition}
    Given $v\in V$, we define the \emph{eccentricity of $v$}, denoted by $\epsilon(v)$, as
    \[ \epsilon(v):= \max_{u\in V} d(u,v).\]
\end{definition}
By definition, the diameter equals the maximum eccentricity and the radius equals the minimum eccentricity. So the problem of computing the diameter consists of finding the maximum over all the eccentricities of the nodes, and similar for the radius.  Le Gall and Magniez~\cite{le2018sublinear} gave two algorithms for diameter computation in the Quantum CONGEST model, a simpler $\bigO(\sqrt{n}D)$ algorithm and a more complicated $\bigO(\sqrt{nD})$ round algorithm. We recover these results below using our framework using a simplified and more general proof.

For a single node, the eccentricity can be computed in $\bigO(D)$ rounds using a Breadth-first search. We immediately obtain a $\bigO(\sqrt{n}D)$ algorithm for computing the diameter or radius by \autoref{cor:thm-with-query-time} and \autoref{lm:batched-minfinding}. To obtain the upper bound of $\bigO(\sqrt{nD})$ we use the following lemma for the classical CONGEST model.
\begin{lemma}\label{lm:eccentricities-of-a-set}
Let $S\subseteq V$ be a set of nodes. In $\bigO(|S|+D)$ classical CONGEST rounds each node in $S$ can learn its of eccentricity.
\end{lemma}
\begin{proof}
    This is a direct consequence from the well known fact that we can construct BFS trees from $|S|$ sources in time $\bigO(|S|+D)$, see for example ~\cite{PRT12,HW12}. 
\end{proof}

The following result retrieves the bound of \cite{le2018sublinear} for computing the diameter, and generalizes to the radius. 
\begin{lemma}
    Given a graph $G=(V,E)$, we can compute the lowest or highest eccentricity in $\bigO(\sqrt{nD})$ rounds in the Quantum CONGEST model with success probability at least $2/3$.
\end{lemma}
\begin{proof}
   We run a parallel minimum finding algorithm, see \autoref{lm:batched-minfinding}, with $p=D$, which implies we have $b=\bigO(\ceil{\sqrt{n/D}})$ for our $(b,p)$-parallel-query quantum algorithm. Now we use \autoref{cor:thm-with-query-time} with $A=[n]$, so $q=\ceil{\log(n)}$. By \autoref{lm:eccentricities-of-a-set} we can compute the eccentricities of $D$ nodes in $\alpha(p)=\alpha(D)=O(D)$ rounds. Hence we get a total number of rounds of
        \begin{align*}
        \bigO\trm{D +b \cdot \trm{(D+p) \ceil{\frac{q}{\log(n)}} + p\ceil{\frac{\log(\tilde{k})}{\log(n)}}+\alpha(p)}} = \bigO\trm{D + \ceil{\sqrt{n/D}}D}= \bigO(\sqrt{nD}).
    \end{align*}
\end{proof}

This technique also generalises to approximating the average eccentricity. 
\begin{lemma}
    Given a graph $G=(V,E)$, we can compute an $\eps$-additive estimate of the average eccentricity in \[\bigO(D+\frac{D^{3/2}}{\epsilon}\cdot\log\left(\sqrt{D}/\epsilon\right)\log\log\left(\sqrt{D}/\epsilon\right))\] rounds, with success probability at least $2/3$.  
\end{lemma}
\begin{proof}
   We use \autoref{lm:parallel-expectation}, where can bound the variance in terms of the diameter: $\sigma^2\leq D^2$. Now if we set $p=D$, we obtain an $\eps$-additive estimate of the average eccentricity with 
    \[b= \bigO(\ceil{\frac{\sigma}{\sqrt{p}\eps} \cdot \log^{3/2}(\frac{\sigma}{\sqrt{p}\eps}) \log\log(\frac{\sigma}{\sqrt{p}\eps})})= \bigO(\ceil{\frac{\sqrt{D}}{\epsilon}\cdot\log\left(\sqrt{D}/\epsilon\right)\log\log\left(\sqrt{D}/\epsilon\right)})\]
    batches of $D$ parallel queries, with success probability at least $2/3$. Now using \autoref{cor:thm-with-query-time} we obtain $\bigO(D+\frac{D^{3/2}}{\epsilon}\cdot\log\left(\sqrt{D}/\epsilon\right)\log\log\left(\sqrt{D}/\epsilon\right))$ rounds.
\end{proof}

\subsection{Finding Cycles of Bounded Length}\label{sec:cycles_of_bounded_length}
In this section, we show a quantum speedup for detecting cycles of length at most $k\geq 4$. We will do this by showing that the classical algorithm of Censor-Hillel et al.~\cite{CFGLLO20} fits in our framework for quantum queries. The algorithm distinguishes two cases for which it finds the cycles separately: light and heavy cycles. We will pick a $\beta$ and call a cycle \emph{light} if all vertices on the cycle have degree at most $n^{\beta}$, if not, we call the cycle \emph{heavy}. We will show a quantum advantage for the finding of heavy cycles. This gives us room to play with in the cost of finding heavy cycles, and by picking $\beta$ smaller than in the original algorithm we decrease the cost of finding light cycles as well, re-balancing the total cost.

\begin{lemma}\label{lm:cyclefinding}
    There is a Quantum CONGEST algorithm that, given $k\geq 4$ and a graph $G=(V,E)$ containing at least one cycle $C_{l}$ of length $l\leq k$, finds such a cycle with smallest $l$ with probability at least $2/3$ in $\bigO(D+(Dn)^{\frac{1}{2}-\frac{1}{4\ceil{k/2}+2)}})$ rounds. 
\end{lemma}
\begin{proof}
Let $\beta$ be a constant we pick later, and let light and heavy cycles be defined as above. 

To find light cycles we use exactly the same procedure as Censor-Hillel et al., which is relatively simple. We restrict ourselves to the graph $G'$ of vertices with degree at most $n^{\beta}$. We run a BFS starting from each node to depth $\delta_k := \ceil{k/2}$. The algorithm halts if a vertex receives the same BFS token twice, indicating it is part of a cycle. Clearly a node finds the smallest cycle it is part of. For the time analysis, note that since we have bounded degree, the $i$-neighborhood of a vertex consists of at most $n^{i\beta}$ vertices. This means that we can run all the BFSs simultaneously in $\bigO\left(\delta_k+n^{\delta_k\beta}\right)=\bigO\left(k+n^{\ceil{k/2}\beta}\right)$ rounds, see \cite{PRT12,HW12}.

Next, we turn to the heavy cycles. This case is slightly more involved. The idea is that we randomly sample a vertex $s\in V$, and find whether $s$ or one of its neighbors is part of cycle of length at most $k$, and if so return the length of the smallest such cycle. We do this as follows: given a vertex $s\in V$, we first run a BFS tree of depth $\delta_k$ from $s$. We set $\kappa$ to be the length of this cycle or $\kappa=k$ if no cycle is found. Then we run separate BFS trees up to depth $\kappa$ on $G\setminus\{s\}$ from each of the neighbors of $s$. Note that we can not simply run a BFS from $s$ with depth $\delta_k+1$, as this might find a cycle of length $k+1$, which is not a valid witness. For a correctness proof we refer to~\cite{CFGLLO20}. 

If in the first step, the BFS from $s$, no cycles are found, then all neighbours of $s$ have a distance of at least $\delta_{\kappa} -1$ from each other when $s$ is removed. Hence, in the second part, each vertex can only receive a BFS token corresponding to one of the neighbours of $s$, except for in the last ($\delta_{\kappa}$-th) round. 
If a vertex ever receives a token it already holds (or if it receives the same token simultaneously from two neighbors) then a cycle is found. If it receives a token different from the one(s) it already holds (necessarily in the last round) than it can safely ignore this token as it does not correspond to a cycle of length $\leq \kappa$, but to a cycle of length $>\kappa$. 

As in the procedure above each vertex is only part of a single BFS, we can run the above procedure for $p$ different vertices $s$ in parallel in $p+\delta_k = \bigO(p+k)$ rounds by \cite{PRT12,HW12}. We run a parallel minimum finding over the vertices, see \autoref{lm:batched-minfinding}, using \autoref{cor:thm-with-query-time}. The value of a vertex is the length of the smallest cycle of length $\leq k$ it or one of its neighbours is a part of, or $\infty$ if there is no such cycle.  If there is a heavy cycle of length $\leq k$ in the network than at least $n^{\beta}$ vertices attain the length of this cycle as a minimum, out of $n$ in total. Plugging everything in the round complexity of \autoref{cor:thm-with-query-time}, with $b=\bigO(\sqrt{\frac{n}{n^{\beta}p}}) = \bigO(n^{\frac{1-\beta}{2}}p^{1/2})$ and $\alpha(p) = \bigO(p+k)$, we get
\[
    \bigO(D+n^{\frac{1-\beta}{2}}p^{-1/2}\trm{D+p+k})
\]
rounds. Here we used that all queries and query results can be stored in $\bigO(\log(n))$ bits. We note that if there is a cycle in the graph, than there is a cycle of length at most $2D+1$. Hence we may set $k\leq 2D+1$ without loss of generality. Now setting $p=D+k = \bigTheta(D)$ gives a round complexity of $\bigO(D+n^{1/2-\beta/2}\sqrt{D})$ for finding a heavy cycle with probability at least $2/3$ if one exists. 

Finally we combine the algorithms for finding light and heavy cycles and picking $\beta = \frac{1+\log_n(D)}{1+2\ceil{k/2}}$ to balance the round complexities we get a total round complexity of
\[
\bigO(D+\sqrt{D}n^{1/2-\beta/2}) = \bigO(D+(Dn)^{1/2-\frac{1}{2+4\ceil{k/2}}}).\qedhere
\]
\end{proof}

Next we use the following lemma from Eden, Fiat, Fischer, Kuhn, and Oshman~\cite{eden2021sublinear} that reduces the problem to a graph with small diameter. This theorem has been used in a similar fashion by Censor-Hillel et al.~\cite{censor2022quantumcliques}. 
\begin{lemma}[Theorem 17 in \cite{eden2021sublinear}]\label{thm:decomp}
    Let $G = (V, E)$ be an $n$-node graph and let $d \geq 2$ be an integer. There is a randomized CONGEST-model algorithm that w.h.p.\ constructs a set of clusters of diameter $O(d \log (n))$ in $O(d\log^2 (n))$ rounds such that each node is in at least one cluster, the clusters are colored with $O(\log (n))$ colors, and clusters of the same color are at distance at least $d$ from each other in $G$. 
\end{lemma}

This gives us a version of \autoref{lm:cyclefinding} without dependency on the diameter. 
\begin{lemma}\label{lm:cyclefinding_noD}
    There is a Quantum CONGEST algorithm that, given $k\geq 4$ and a graph $G=(V,E)$ containing at least one cycle $C_{l}$ of length $l\leq k$, finds such a cycle with probability at least $2/3$ in $$\bigO(\left(k+(kn)^{\frac{1}{2}-\frac{1}{4\ceil{k/2}+2)}}\right)\log^2(n))$$ rounds. 
\end{lemma}
\begin{proof}
    We apply \autoref{thm:decomp} with $d=2k$ to obtain a clustering as stated. Now for each cluster we elect a leader in $O(k\log^2(n))$ rounds. Next for each color $i$, we run the algorithm from \autoref{lm:cyclefinding} on all of the clusters with color $i$,  go up to distance $k$ outside the cluster. The diameter of the subgraph we consider therefore is $O(k+k\log(n))=O(k\log(n))$ and we find any cycle of length at most $k$ if it has at least one vertex in the cluster. Moreover, the algorithms operate on disjoint nodes as the cluster from one color are at least $2k$ apart. This computation takes $\bigO(k\log(n)+(kn\log(n))^{\frac{1}{2}-\frac{1}{4\ceil{k/2}+2)}})$ rounds per color, so $\bigO(k\log^2(n)+\log(n)(kn\log(n))^{\frac{1}{2}-\frac{1}{4\ceil{k/2}+2)}})$ rounds in total. 
    
    Regarding the error probability, as we apply our previous algorithm multiple times, we note that we can efficiently check whether a found cycle indeed exists, so we will never output a cycle of length smaller than the smallest cycle. Hence, the only possible error is the missing of all cycles of minimal length and outputting a larger cycle or ``no cycle found''. In particular this requires an error to occur on a cluster with a cycle of minimal length, which has probability $\leq 1/3$, and hence the total error probability is $\leq 1/3$.
\end{proof}

Using the same technique, we can implement a quantum algorithm for detecting cycles of \emph{exactly} length $k=4,6,8,10$ in time $\bigO(n^{\frac{1}{2}-\frac{1}{2k+2}})$. This uses the classical algorithm from Censor-Hillel et al.~\cite{CFGLLO20} for computing small even cycles as a basis. Their algorithm works using so-called color-BFSs in place of the BFSs in the above algorithm. This would show a separation for this problem between quantum and classical CONGEST, as detecting even cycles of any length requires $\bigOmegat(\sqrt{n})$ rounds in classical CONGEST~\cite{KR18}.

\subsection{Computing the Girth}
In this section, we will show a quantum speedup for computing the girth of a graph. The goal is that at the end of the algorithm the minimum output among all nodes is the girth. To make this global knowledge an additional $D$ rounds suffice. Our algorithm is based upon repeatedly searching for cycles of length at most $k$, analogous to the classical algorithm of Censor-Hillel et al.~\cite{CFGLLO20}. Triangles are a special case, and will be treated separately. 
Now to find the girth we repeat this in a binary search fashion. 
\begin{corollary}
    Given $\mu>0$, there is a Quantum CONGEST algorithm to compute the girth $g$ of a graph $G=(V,E)$ in $\bigO(\frac{1}{\mu}\left(g+\left(gn\right)^{\frac{1}{2}-\frac{1}{4\ceil{g(1+\mu)/2}+2}}\right)\log^2(n))$ rounds with success probability at least $2/3$. No upper bound on $g$ needs to be known in advance.
\end{corollary}
\begin{proof}
   First try to find a triangle, which can be done in $\bigOt(n^{1/5})$ rounds in Quantum CONGEST \cite{censor2022quantumcliques}. Next, for $0<\mu\leq 1$, try to find a cycle of length at most $k=4,4(1+\mu),4(1+\mu)^2,\dots$, which terminates in $\bigO(\log(g)/\log(1+\mu))$ iterations, say at $k=4(1+\mu)^j$. By \autoref{lm:cyclefinding_noD}, we immediately find a cycle of length $g$ if $g\leq k$ and the algorithm does not error in this round, which happens with probability at least 2/3. As we are again dealing with one-sided error the algorithm never stops prematurely. We see that the algorithm takes at most \begin{align*}
       &\bigO(\sum_{i=0}^{\ceil{\log_{1+\mu}(g/4)}}\left(4(1+\mu)^i\log^2(n)+\log(n)\left(4(1+\mu)^i n\log(n)\right)^{\frac{1}{2}-\frac{1}{4\cdot\lceil2(1+\mu)^i\rceil+2}}\right))\\
   &= \bigO(\frac{1}{\mu}\left(g+\left(gn\right)^{\frac{1}{2}-\frac{1}{4\ceil{g(1+\mu)/2}+2}}\right)\log^2(n))
   \end{align*}
   rounds. 
\end{proof}   
In the classical CONGEST model, computing a $(2-\epsilon)$-approximations of the girth of a graph takes at least $\Omega(\sqrt{n})$ rounds \cite{FHW12}. Since this also holds for graphs with constant girth, we have therefore shown a separation between the quantum and classical CONGEST model for girth computation in low girth graphs. If in the end we want the whole graph to know the girth, we get an additive factor $D$ in our running time. The lower bound in this case also becomes $\Omega(\sqrt{n}+D)$.

\section{Non-Oracle Techniques}
Finally, in this section we provide a short discussion on quantum algorithmic techniques that do not fit in our framework directly, but can still be implemented in the quantum CONGEST model. In particular, our framework requires that the queries to the network act exactly as a query would in a quantum algorithm, and hence do not allow for any left over ancillary registers at the nodes. We focus on amplitude amplification, phase estimation, and amplitude estimation, all of which apply to some black-box (quantum) subroutine that might not be a nice query. 

\begin{lemma}[Amplitude Amplification Iterate]\label{lm:ampamp-itt}
  Consider a state
  \[
  \ket{\psi} = \sqrt{1-p} \ket{\phi_0}\ket{0}+\sqrt{p} \ket{\phi_1}\ket{1}
  \]
  for some $p\in[0,1]$.
  Let $U_{\ket{\psi}}$ be the unitary corresponding to a $R_{\ket{\psi}}$-round quantum CONGEST algorithm that prepares $\ket{\psi}$ shared by the nodes of the network from the all-zero state. Then
  \[
  \bigO(R_{\ket{\psi}} +D) 
  \]
  rounds suffice to implement the amplitude amplification iterate.
\end{lemma}
\begin{proof}
  The amplitude amplification iterate is the product of two reflections: through the space spanned by $\ket{\psi}$ and through the space spanned by the ``Good'' part $\ket{\phi_1}\ket{1}$. The reflection through the ``Good'' part can be implemented with a single $Z$ gate by the node that holds the last register.

  The reflection through $\ket{\psi}$ takes a slightly more careful approach. As in the normal quantum algorithmic implementation, we first apply $U_{\ket{\psi}}^{\dagger}$. This allows us to reflect through $\ket{0}$ instead of $\ket{\psi}$. To do so we need to recognize the all zero state, but the full state of the algorithm can be shared by the nodes and might not be located at a single leader node. We can avoid collecting the full state at the leader by letting each node check whether their local registers are all zero and computing the AND of the result at a leader, requiring $\bigO(D)$ rounds. The leader can then apply a $Z$ gate to implement the reflection. Undoing the computation and applying $U_{\ket{\psi}}$ completes the reflection through the span of $\ket{\psi}$.
\end{proof}

As a result we directly get the following corollary. 
\begin{corollary}[Amplitude Amplification]\label{cor:amplitude_amplification}
  Consider the setting of \autoref{lm:ampamp-itt}.
  \[
  \bigO(\trm{R_{\ket{\psi}} +D}\frac{1}{\sqrt{p}}\log(1/\delta))
  \]
  rounds suffice to obtain $\ket{\phi_1}$ with success probability at least $1-\delta$.
\end{corollary}
\begin{proof}
This follows from the normal proof of amplitude amplification, with the additional note that we can check if we obtained $\phi_1$ (and communicate this to all nodes in the stated number of rounds) and hence can use $\bigO(\log(1/\delta))$ repetitions to boost the success probability.
\end{proof}
This amplitude amplification algorithm can be used to boost the success probability of randomized CONGEST algorithms, even when the random coinflips are performed in many different places in the network. Applying amplitude amplification using our framework from \autoref{sc:mainthm} would require the leader to flip all of the coins and communicate this to the network, which might lead to congestion problems. 

Possibly more interesting is a distributed version of phase estimation. As phase-kickback does not require any communication, the diameter dependence only comes in at the start and end of the algorithm and hence is fully additive.
\begin{lemma}[Phase Estimation]\label{lem:phaseest}
  If the nodes share a state $\ket{\psi}$ and can apply a unitary $U$ using $R$ rounds such that $U\ket{\psi} = e^{\ii \theta}\ket{\psi}$, then
    \[
\bigO(\frac{R}{\eps}\log(1/\delta) +D)
\]
rounds suffice for the leader to learn $\theta$ up to additive error $\eps$ with success probability at least $1-\delta$.
\end{lemma}
\begin{proof}
  We first consider the case where $\delta = 1/3$. The algorithm follows normal phase estimation: the leader prepares a superposition over numbers $k = 1,\dots,\bigO(1/\eps)$, shares this with the network using \autoref{lm:The-Van-Apeldoorn-De-Vos-Lemma} in $\bigO(D+\log(1/\eps))$ rounds. The network then performs $U^k$ conditioned on the value of $k$ send by the leader. Note that there are no additional $D$ steps here (although $R$ might have a $D$-dependence, depending on the application). The nodes then un-share the superposition over $k$ and the leader applies an inverse Quantum Fourier Transform. The correctness follows form the correctness of standard phase estimation.
  
  To boost the success probability, the algorithm above is repeated $\bigO(\log(1/\delta))$ times and the leader takes the median value of the result. Note that the $k$ values for the different runs can be streamed at the same time using \autoref{lm:The-Van-Apeldoorn-De-Vos-Lemma}. 
\end{proof}

As amplitude estimation is simply phase estimation applied to the amplitude amplification iterate, we get the following corollary.
\begin{corollary}[Amplitude Estimation]
  Consider the setting of \autoref{lm:ampamp-itt} and assume we are given a constant $p_{max}$ such that $p\leq p_{max}$. Then
  \[
\bigO(\trm{R_{\ket{\psi}} + D} \frac{\sqrt{p_{max}}}{\eps}\log(1/\delta))
\]
rounds suffice in the Quantum Congest model to estimate $p$ up to additive error $\eps$ with success probability at least $1-\delta$.
\end{corollary}
\begin{proof}
This follows from applying \autoref{lem:phaseest} to \autoref{lm:ampamp-itt}. Note that phase estimation estimates the rotation angle of the amplitude estimation iterate, not the probability directly. Due to this conversion only $\frac{\sqrt{p_{max}}}{\eps}$ iterations are necessary, as shown in \cite{brassard:ampest}.
\end{proof}

\begingroup
\emergencystretch 2em
\sloppy
\hbadness 10000
\printbibliography
\endgroup

\end{document}